\documentclass[11pt,fleqn]{article}

\usepackage{amssymb,amsmath}
\usepackage{natbib}
\usepackage{vmargin,setspace,verbatim}
\usepackage{graphicx,graphics,color,tikz,tikz-network}
\usepackage{microtype}

\usepackage{libertine}
\usepackage[libertine]{newtxmath}

\definecolor{darkgreen}{rgb}{0,0.2,0}
\definecolor{darkred}{rgb}{0.3,0,0}

\usepackage[colorlinks=true, pdfstartview=FitV, linkcolor=darkgreen, citecolor=darkred, urlcolor=black]{hyperref}

\newcounter{llst}
\newenvironment{abet}{\begin{list}{\rm (\alph{llst})}{\usecounter{llst}
\setlength{\itemindent}{0em} \setlength{\leftmargin}{3em}
\setlength{\labelwidth}{2em} \setlength{\labelsep}{1em}}}{\end{list}}
\newenvironment{numm}{\begin{list}{\rm (\roman{llst})}{\usecounter{llst}
\setlength{\itemindent}{0em} \setlength{\leftmargin}{3.5em}
\setlength{\labelwidth}{2.5em} \setlength{\labelsep}{1em}}}{\end{list}}

\newtheorem{theorem}{Theorem}[section]

\newtheorem{axiom}[theorem]{Axiom}

\newtheorem{corollary}[theorem]{Corollary}

\newtheorem{definition}[theorem]{Definition}
\newtheorem{expl}[theorem]{Example}

\newtheorem{proposition}[theorem]{Proposition}
\newtheorem{remrk}[theorem]{Remark}

\newtheorem{dscrpt}[theorem]{Description}

\newcounter{axiomatiser}

\newcounter{myclaimcount}
\newenvironment{claim}{\stepcounter{myclaimcount} \vspace*{1ex} \noindent \textbf{Claim \Alph{myclaimcount}: } \it}{\vspace*{1ex}}

\newenvironment{proof}[1][Proof]{\noindent \textbf{#1.} }{\hfill
\rule{0.5em}{0.5em}}

\newenvironment{example}{\begin{expl} \rm}{\hfill $\blacklozenge$
\end{expl}}

\newenvironment{remark}{\begin{remrk} \rm}{\hfill $\blacklozenge$
\end{remrk}}

\setpapersize{A4}
\setmarginsrb{80pt}{40pt}{80pt}{60pt}{12pt}{10pt}{12pt}{30pt}

\begin{document}


\title{\textbf{Emergent Collaboration  in Social Purpose Games}\thanks{We thank Stef Tijs, Emiliya Lazarova and Subhadip Chakrabarti for their inspirational leadership in the partially cooperative development of the equilibrium concepts used in this paper.}}

\author{Robert P.~Gilles\thanks{Management School, The Queen's University of Belfast, Riddel Hall, 185 Stranmillis Road, Belfast, BT9~5EE, UK. \textsf{Email: r.gilles@qub.ac.uk}} 
 \and
 Lina Mallozzi \thanks{Department of Mathematics and Applications, University of Naples Federico II, Via Claudio 21, 80125 Naples, Italy. E-mail: \textsf{mallozzi@unina.it}} 
 \and
 Roberta Messalli\thanks{Department of Economic and Statistical Science, University of Naples Federico II, Complesso Monte Sant'Angelo 21, 80125 Naples, Italy. E-mail: \textsf{robertamessalli@gmail.com}}
 }

\date{September 2021}

\maketitle

\begin{abstract}
\singlespace
\noindent
We study a class of non-cooperative aggregative games---denoted as \emph{social purpose games}---in which the payoffs depend separately on a player's own strategy (individual benefits) and on a function of the strategy profile which is common to all players (social benefits) weighted by an individual benefit parameter. This structure allows for an asymmetric assessment of the social benefit across players. 

We show that these games have a potential and we investigate its properties. We investigate the payoff structure and the uniqueness of Nash equilibria and social optima. Furthermore, following the literature on partial cooperation, we investigate the leadership of a single coalition of cooperators while the rest of players act as non-cooperative followers. In particular, we show that social purpose games admit the emergence of a stable coalition of cooperators for the subclass of \emph{strict} social purpose games. Due to the nature of the partial cooperative leadership equilibrium, stable coalitions of cooperators reflect a limited form of farsightedness in their formation.

As a particular application, we study the \emph{tragedy of the commons game}. We show that there emerges a single stable coalition of cooperators to curb the over-exploitation of the resource.

\end{abstract}

\begin{description}
\singlespace
\item[Keywords:] Partial cooperation; Leadership equilibrium;  Potential games;  Aggregative games;  Tragedy of the Commons.

\item[JEL classification:] C71, D71
\end{description}

\thispagestyle{empty}

\newpage

\setcounter{page}{1} \pagenumbering{arabic}

\section{Introduction: Endogenous partial cooperation in games}

In this paper, following the literature of additively separable aggregative games \citep{Corchon1994,Dubey2006} and in line with the asymmetry considered in \citet{McGuinty2013}, we introduce a class of non-cooperative games---denoted as \emph{Social Purpose Games\/}---for which the payoff of each player depends separately on his own strategy and on a function of the strategy profile, the aggregation function, which is the same for all players, weighted by an individual benefit parameter which enlightens the asymmetry, between agents, of public benefit. The two parts of the payoff function represent respectively the individual and the social benefits. 

The purpose of this paper is to investigate the properties of and relationship between Nash equilibria and social optima in social purpose games.   Furthermore, we show that for a subclass of these social purpose games there emerges a stable coalition of cooperators founded on the presence of a limited form of farsightedness in the corresponding definition of stability. 

The main feature of social purpose games is that they capture the tension between social benefits and individual payoffs. As such this class of social purpose games includes implementations of the tragedy of the commons, pollution abatement games, and public good provision games. In most of these social purpose games, individual optimality leads players to underutilise social benefits in favour of direct individualistic payoffs. We show that, indeed, this is a general feature of social purpose games, exemplified by the social suboptimality of the Nash equilibria. If we interpret the strategies in these social purpose games to represent a chosen level of ``effort'', the total Nash equilibrium generated level of effort is socially suboptimal. We argue that partial collaboration through the emergence of a stable coalition of cooperators, partially abates the suboptimality of the resulting effort levels. Therefore, this fundamental insight holds for a substantial class of aggregative games. 

For our analysis we introduce two subclasses of social purpose games. In a \emph{regular} social purpose game, all payoff functions are continuous and (weakly) concave. Regularity of a social purpose game guarantees the existence of the main equilibrium concepts. A social purpose game is \emph{strict} if the payoff structures are strictly convex, guaranteeing uniqueness of the main equilibrium concepts.

\citet{Jensen2010} already pointed out that classes of aggregative games can be characterised through certain forms of potential functions. We show here that social purpose games admit a weighted potential function. This is a substantially stronger property than Jensen showed for certain classes of aggregative games. Furthermore, the existence of a weighted potential leads to the question regarding the relationship between the maximisers of the potential function and the Nash equilibria of these games. We show that for a special class of social purpose games the set of Nash equilibria actually coincides with the set of potential maximisers.

Concerning the Nash equilibria of social purpose games, we show a wide range of properties. Under standard continuity properties, any social purpose game admits at least one Nash equilibrium, including the class of regular social purpose games. For the class of \emph{strict} social purpose games the Nash equilibrium is unique due to the strict convexity of the corresponding payoff structures. 

Similarly, under these same conditions, social purpose games generally admit at least one social optimum, while for the subclass of strict social purpose games the social optimum is unique. The fundamental tension between social and individual benefit leads to the natural conclusion that, as stated above, for a large subclass of social purpose games, the Nash equilibrium levels of effort are socially suboptimal. The next goal of our paper is to investigate how partial cooperation can alleviate tension between socially optimal and equilibrium levels of effort. 

\paragraph{Partial cooperation in social purpose games}

Social purpose games have a specific structure that allows the emergence of stable partial cooperation among players. Our investigation is rooted in the work on partial cooperation in a range of types of non-cooperative games by \citet{Chander1997}, \citet{MallozziTijs2006,MallozziTijs2007,MallozziTijs2009} and \citet{CGL2011,CGL2018}. In this literature, one considers the formation of a single ``coalition of cooperators'' that collectively determines a joint strategy to maximise its collective payoffs---being the sum of the individual members' payoffs. This is akin to cartel formation in oligopolistic market games \citep{Daspremont1983} and international treaty writing in the context of environmental abatement situations \citep{Diamantoudi2006,Kwon2006,Olmstead2014}. There are two natural behavioural hypotheses that can be considered in this context.

First, one can assume that the coalition of cooperators acts as a single player under standard best response rationality. The resulting stable strategy profiles are referred to as \emph{partial cooperative equilibria}, which existence can be established under relatively mild conditions \citep{CGL2011,CGL2018}. This avenue is not investigated in our analysis, since in many applications the coalition of cooperators would not have a standard position in the corresponding decision making processes; in most natural applications, the coalition of cooperators is found to have a position of \emph{leadership} in a decision-making hierarchy.

The hypothesis, that the coalition of cooperators assumes a leadership position as a first mover in relation to the non-cooperating players in the game, was seminally proposed by \citet{Stackelberg1934} and \citet{Chander1997}. This idea has been developed further as the notion of a \emph{partial cooperative leadership equilibrium} (PCLE) in the cited literature on partial cooperation in general non-cooperative games. It is clear that the leadership position of the coalition of cooperators gives it an advantage in comparison to the standard partial cooperative equilibrium payoff. Nevertheless, this leadership position seems natural and has been observed in the context of social purpose games, in particular for public good provision games, pollution abatement games, and (oligopolistic) market games. 

In our study of partial cooperation in social purpose games, we limit ourselves to the investigation of PCLE under the formation of a coalition of cooperators. In particular, we can show that under relatively mild conditions there exists a unique PCLE in a social purpose game. This conception, therefore, allows us the introduction of farsightedness in coalition formation in these games. Namely, we assume that if a coalition of cooperators emerges in a social purpose game, it assumes a leadership position and there naturally emerges a unique PCLE. 

This affords us with the possibility to consider the emergence of a ``stable'' coalition of cooperators. The applied notion of stability in coalition formation is founded on the von Neumann-Morgenstern standards \citep{vNM}, which considers two partial stability properties: (1) \emph{Internal stability}---every cooperator will obtain a lower payoff upon leaving the coalition of cooperators, and (2) \emph{External stability}---every non-cooperator receives a lower payoff upon joining the existing coalition of cooperators. 

We show for a subclass of the strict social purpose games there indeed emerge stable coalitions of cooperators that abate the social suboptimality of the corresponding equilibrium effort levels. 

\paragraph{Relationship to the literature}

An aggregative game is founded on the hypothesis that each payoff function depends on the corresponding player's strategy as well as on some aggregation of all selected strategies. Classical examples of aggregation are the unweighted sum and the mean.  The concept of aggregative games goes back to \citet{Selten1970}, who considers as aggregation function the summation of all the players' strategies. Later, this concept has been studied in the case of other aggregation functions and it has been generalised to the concept of quasi-aggregative games. For this we refer to \citet{Vives1990}, \citet{Corchon1994}, \citet{Cornes2005}, \citet{Dubey2006}, \citet{Miguel2009}, \citet{Jensen2010}, \citet{Jensen2013}, and  \citet{MallozziMessalli2017}. In this literature, there are many games that present an aggregative structure: among them, we mention Cournot and Bertrand games, patent races, models of contests of fighting and model with aggregate demand externalities. 

In additively separable aggregative games, each payoff function is a sum of a function that depends on an aggregation of strategies and a function that depends on player's own strategy. The model of additively separable aggregative games appeared in literature, among others, in the context of International Environmental Agreements (IEA), studying the formation of stable IEA in the case in which each country's choice variable is emission and then extending the results to the dual case, i.e., the case where the choice variable is abatement effort \citep{Diamantoudi2006}.

Also public good provision games are in the context of additively separable aggregative games \citep{BergstromBlumeVarian1986} where each player consumes a certain amount of a private good and donates a certain other amount to the supply of the public good. Thus, the payoff function of each player turns to depend not only on the quantity of private good that he consumes but also on all contributions to the public good made by all individuals. \citet{McGuinty2013} investigate the impact of asymmetry in a voluntary public goods environment by proposing an improved design that explicitly isolates individual incentives, without assuming a dominant strategy. 

\paragraph{Structure of the paper}

In Section 2 of this paper we investigate thoroughly the quintessential social purpose game, namely a standard implementation of the tragedy of the commons \citep{Hardin1968}. We show that the effects of the over-exploitation of the commons can be mitigated by partial cooperation among the players and that such cooperation can be stable. Stable partial cooperation is shown to mitigate the over-exploitation of the commons.

In Section 3 we introduce the class of general social purpose games and identify some relevant sub-classes. We show that social purpose games in general admit a weighted potential and that for a sub-class of these games the potential maximisers coincide with the set of Nash equilibria. We discuss social optima and their properties, showing in particular that for the class of regular social purpose games there is the expected relationship between Nash equilibrium and socially optimal strategies, referring to the over-exploitation of the commons. 

Section 4 introduces the notion of partial cooperation in social purpose games. We define the notion of a partial cooperative leadership equilibrium and prove its existence for regular social purpose games. Subsequently, we introduce a stability concept in the formation of partial cooperation and show that it is natural to expect that strict social purpose games admit the emergence of stable partial cooperation. This is fully developed for an application to quadratic payoff functions.

\section{A motivating case: The tragedy of the commons}

The ``tragedy of the commons'' refers to the classical problem of the overuse and exploitation of a common resource through free-riding in a non-cooperative setting. Traditionally the common resource referred to a common tract of land in a medieval village for the grazing of cattle owned by the village's peasants. If peasants freely access and use the land, a situation of over-exploitation arises, resulting in the depletion of the commons for collective use \citep{Lloyd1833}.\footnote{The tragedy of the commons has significant appeal and application in our contemporary global economy, referring to contemporary issues such as the exploitation of natural resources---including fish stocks, fresh water sources, mining of ores, and oligopolistic commodity markets---as well as the global environmental conditions. The tragedy of the commons is known as one of the most fundamental examples of a \emph{social dilemma}.}

The tragedy of the commons was formulated in game-theoretic terms by \citet{Hardin1968} and has been considered a totemic reference in many contributions to the social and biological sciences \citep{Frischmann2019}. \citet{Hardin1968} only considered the non-cooperative case of unlimited and free extraction from the commons. \citet{Ostrom1990} challenged Hardin's reductionism by pointing out that in many social situations the commons was and still is successfully governed through the application of institutional and behavioural solutions. \citet{SDL2020-JEBO} introduce a mathematical model of an institutional solution on managing the commons from a public good perspective based on Ostrom's insights.

Here we pursue a third perspective on the management of the commons by investigating the endogenous emergence of a coalition of cooperators that collective regulate their extraction from the commons, while non-members of this coalition of cooperators selfishly extract. We can show that the extraction from the commons under such \emph{partial cooperation} significantly improves collective wealth generation and welfare. We illustrate the benefits from endogenous partial cooperation in the classical tragedy of the commons game by considering a simple example.

\paragraph{A non-cooperative extraction game}

Suppose that there is a finite commonly owned resource that has a total size of one (1). There are $n \in \mathbb N$ with $n \geqslant 3$ users of this resource, who are individually free to extract any benefit from the common resource with the understanding that any future benefit from the resource would be limited by the extent of today's usage. Hence, future use is based on the remainder of the common resource at the conclusion of today's collective extraction.

This results in a standard non-cooperative game $\langle \, N, (S_i)_{i \in N} , (\pi_i)_{i \in N} \, \rangle$ where $N = \{ 1,2, \ldots ,n \}$ is the set of users and $S_i = [0,1]$ is the set of individual extraction levels. We apply a standard payoff function that considers benefits from current usage and the utility of future usage with equal weight for every $i \in N$. Hence,
\[
\pi_i (x) = x_i \, \left( 1 - \sum_{j \in N} x_j \right) = x_i \left( 1-X_N \right) \qquad \mbox{for every } x \in S = \prod_{j \in N} S_j = [0,1]^n \subset \mathbb R^n .
\]
We denote by $X_T = \sum_{i \in T} x_i$ the total extraction of a coalition $T \subseteq N$ with the convention that $X_{\varnothing} =0$. Note that $\left( 1-X_N \right)$ therefore represents the total size of the commons left for future usage.\footnote{We remark that $1-X_N < 0$ refers to the destruction of the commons, resulting in negative payoffs for all users.}

\paragraph{Nash equilibria and social optima}

We summarise the resulting Nash equilibrium\footnote{A Nash equilibrium is a strategy profile $x^{\mathrm{NE}} \in S$ in which all users do a best response to what other users extract from the commons, i.e., $\pi_i (x^{\mathrm{NE}}) = \max_{x_i \in S_i} \pi_I (x_i, x^{\mathrm{NE}}_{-i} )$.} and the unique social (Pareto) optimum\footnote{A social optimum is a strategy profile $x^{\mathrm{SO}} \in S$ that maximizes the collective welfare, i.e., $\sum_{i \in N} \pi_i (x^{\mathrm{SO}}) = \max_{x \in S} \sum_{i \in N} \pi_i (x)$. } outcomes in the following table:

\begin{center}
\begin{tabular}{|l|cc|cc|}
\hline
 & $x_i$ & $X_N$ & $\pi_i$ & $\sum \pi_i$ \\
\hline
Nash equilibrium
& $\frac{1}{n+1}$ & $\frac{n}{n+1}$ & $\frac{1}{(n+1)^2}$ & $\frac{n}{(n+1)^2}$ \\
\hline
Social optimum
& $\frac{1}{2n}$ & $\frac{1}{2}$ & $\frac{1}{4n}$ & $\frac{1}{4}$ \\
\hline
 \end{tabular}
\end{center}

\noindent
We observe here the well-known conclusion that there significant over-extraction in the Nash equilibrium, which reduces welfare uniformly for all users. Indeed, we note that $X^{\mathrm{NE}}_N > X^{\mathrm{SO}}_N$ as well as $\sum \pi^{\mathrm{NE}}_i < \sum \pi^{\mathrm{SO}}_i$ for all $n \geqslant 3$.

\paragraph{Partial cooperation in the tragedy of the commons}

Next we consider a hybrid of best response rationality and socially optimal decision-making by allowing users to collaborate to extract in a collectively rational fashion. In particular, consider that a coalition of cooperators $C \subset N$ decides collectively about their extraction rate, while the non-collaborating users $j \in N \setminus C$ act according to best response rationality. Furthermore, we assume that the coalition of cooperators $C$ assumes a Stackelberg leadership position and acts as a first-mover. Hence, users in $C$ determine their coordinated extraction rates taking into account the best responses of all non-cooperators $j \in N \setminus C$. This is referred to as a \emph{partial cooperative leadership equilibrium} (PCLE) in \citet{CGL2011,CGL2018}.

Suppose that the coalition of cooperators has size $|C|= m < n$. In a PCLE, non-cooperators $j \in N \setminus C$ optimise their payoffs given $X_C$ as well as $(x_k)_{k \in N \setminus C , \, k \neq j}$. This results for the non-cooperator $j \in N \setminus C$ in solving
\[
\max_{0 \leqslant x_j \leqslant 1} \, \pi_j \left( x_j, x^{NC},x^C \right) = x_j \left ( 1 - X_C - X_{NC} \right)
\]
where $X_C = \sum_{i \in C} x_i$ and $X_{NC} = \sum _{h \in N \setminus C} x_h$. Solving this simultaneously for all $j \in N\setminus C$, given $X_C$, we conclude that the best response for $j \in N \setminus C$ is given by
\begin{equation}
	x_j (X_C) = \frac{1-X_C}{n-m+1} \quad \mbox{and} \quad \sum_{j \in N \setminus C} x_j (X_C) =\frac{(n-m) \, \left( 1-X_C \right)}{n-m+1}
\end{equation}
Next, the coalition of cooperators $C$ collectively determine their collective extraction, given the optimal decisions of all non-cooperators, by solving
\[
\max_{x^C \in [0,1]^C} \, \sum_{i \in C} \pi_i \left( x^C, \left( x_j (X_C) \right)_{j \in N \setminus C} \right) = X_C \left( 1- X_C - \sum_{j \in N \setminus C} x_j (X_C) \, \right)
\]
This results in the conclusion that the coalition of cooperators $C$ solves
\begin{equation}
	\max_{x^c \in [0,1]^C} \, X_C \cdot \frac{1-X_C}{n-m+1}
\end{equation}
leading to the conclusion that the optimal collective extraction is $X^L_C = \frac{1}{2}$. The resulting PCLE can be summarised as
\begin{align}
	x^L_i & = x^L_C (m) = \frac{1}{2m} & \mbox{for } i \in C \\
	x^L_j & = x^L_{NC} (m) = \frac{1}{2(n-m+1)} & \mbox{for } j \in N \setminus C 
\end{align}
resulting in $X^L_C (m) = \tfrac{1}{2}$, $X^L_{NC} (m) = \tfrac{n-m}{2(n-m+1)}$ and $X^L_N (m) = 1- \tfrac{1}{2(n-m+1)}$. Clearly, the resulting PCLE outcomes form a hybrid between Nash equilibrium and social optimum extraction. The resulting PCLE payoffs reflect this as
\begin{align}
	\pi^L_i & = \pi^L_C (m) = \frac{1}{4m(n-m+1)} & \mbox{for } i \in C \\
	\pi^L_j & = \pi^L_{NC} (m) = \frac{1}{4(n-m+1)^2} & \mbox{for } j \in N \setminus C
\end{align}

\paragraph{Stable partial cooperation}

The tragedy of the commons as formulated here allows for the consideration of stability in the process of partial cooperation under leadership. We call a coalition of cooperators $C$ \emph{stable} if no non-cooperator $j \in N \setminus C$ would like to join the coalition $C$ and no member $i \in C$ would like to leave the coalition of cooperators. Note that this introduces a model of coalition formation that is founded on one-step farsightedness on part of all users in this tragedy of the commons situation. 

Formally, we say that a coalition of cooperators $C \subseteq N$ is stable if
\begin{equation}
	\pi^L_{NC} (m) \geqslant \pi^L_C (m+1) \qquad \mbox{as well as} \qquad \pi^L_C (m) \geqslant \pi^L_{NC} (m-1)
\end{equation}
After some calculus we conclude that a coalition of cooperators $C$ is stable if and only if
\[
\frac{m+1}{n-m+1} \geqslant \frac{n-m+1}{n-m} \qquad \mbox{as well as} \qquad \frac{n-m+2}{n-m+1} \geqslant \frac{m}{n-m+2}
\]
It can be verified that there are many $(n,m)$ pairs that satisfy these inequalities. In particular, this holds for $n=8$ and $m=5$. We verify this for these particular sizes of the coalition of cooperators for population size $n=8$. The next table summarises the main resulting PCLE strategies and extraction rates for all sizes $2 \leqslant m <n=8$ of the coalition of cooperators. 

\begin{center}
\begin{tabular}{|l|ccc|cc|}
\hline
$m$ & $x^L_C$ & $x^L_{NC}$ & $X_N$ & $\pi^L_C$ & $\pi^L_{NC}$
\\
\hline
2 & $\tfrac{1}{4}$ & $\tfrac{1}{14}$ & $\tfrac{13}{14}$ & $\tfrac{1}{56}$ & $\tfrac{1}{196}$
\\
3 & $\tfrac{1}{6}$ & $\tfrac{1}{12}$ & $\tfrac{11}{12}$ & $\tfrac{1}{72}$ & $\tfrac{1}{144}$
\\
4 & $\tfrac{1}{8}$ & $\tfrac{1}{10}$ & $\tfrac{9}{10}$ & $\tfrac{1}{80}$ & $\tfrac{1}{100}$
\\
$5^*$ & $\tfrac{1}{10}$ & $\tfrac{1}{8}$ & $\tfrac{7}{8}$ & $\tfrac{1}{80}$ & $\tfrac{1}{64}$
\\
6 & $\tfrac{1}{12}$ & $\tfrac{1}{6}$ & $\tfrac{5}{6}$ & $\tfrac{1}{72}$ & $\tfrac{1}{36}$
\\
7 & $\tfrac{1}{14}$ & $\tfrac{1}{4}$ & $\tfrac{3}{4}$ & $\tfrac{1}{56}$ & $\tfrac{1}{16}$
\\
\hline
\end{tabular}
\end{center}

\noindent
From the table of the computed PCLE, we conclude that indeed $m^*=5$ is a stable size of the coalition. Indeed, $\pi^L_C (5) = \tfrac{1}{80} > \pi^L_{NC} (4) = \tfrac{1}{100}$ as well as $\pi^L_{NC} (5) = \tfrac{1}{64} > \pi^L_C (6) = \tfrac{1}{72}$. 

The analysis of the PCLE in the setting of the tragedy of the commons shows that there emerges self-regulation of the extraction from the commons through partial cooperation under a limited farsighted behavioural rationality. In this paper we explore a more general class of non-cooperative games that exhibits this feature. This is developed over the following sections. 

\section{Social purpose games}

In the theory of non-cooperative games, a \emph{normal form game} is an interactive decision situation that is represented as a list $G = \left \langle {N, (S_i)_{i=1}^n, ({\pi}_i)_{i=1}^n} \right \rangle$ with
\begin{itemize}
\item $N = \{ 1, \ldots ,n \}$ is a given finite set of players, where $n \in \mathbb{N}$ is the number of players;
\item for every $i \in N$, $S_i$ is a non-empty strategy or action set for player $i$, where we define $S = \prod_{i \in N} S_i$ as the set of all strategy tuples in $G$, and;
\item for every $i \in N$, $\pi_i \colon S \to \mathbb{R}$ is the payoff function of player $i$.
\end{itemize}
A \emph{Nash equilibrium} in the game $G= \left \langle {N, (S_i)_{i=1}^n, ({\pi}_i)_{i=1}^n} \right \rangle$ is a strategy tuple $x^* \in S$ such that for every $i \in N \colon \pi_i (x^*) \geqslant \pi_i \left( x^*_{-i} , x_i  \right)$ for every alternative strategy $x_i \in S_i$, where the tuple of strategies of all players except $i$ is denoted by $x_{-i} = \left( x_1 , \ldots , x_{i-1} , x_{i+1} , \ldots , x_n \right) \in \prod_{j \neq i} S_j$.

We now introduce a social purpose game as a non-cooperative normal form game with a specific payoff structure. 

\begin{definition}
A \textbf{social purpose game} is a list $\Gamma = \langle N , \overline{Q}, H, ( \alpha_i , h_i,g_i)_{i \in N} \rangle$ that defines an associated non-cooperative game $G_\Gamma = \left \langle {N, (S_i)_{i=1}^n, ({\pi}_i)_{i=1}^n} \right \rangle$ with
\begin{itemize}
\item $N = \{ 1, \ldots ,n \}$ is a finite player set, where the number of players is $n \in \mathbb{N}$ such that $n \geqslant 2$;

\item$\overline{Q}>0$ is a positive number such that every $i \in N \colon S_i = [ \, 0,\overline{Q} \, ]$, endowed with the standard Euclidean topology;

\item $H\colon\mathbb{R} \to \mathbb{R}$ is a function, and;

\item For every $i \in N \colon \alpha_i > 0$ is a parametric weight and $h_i, g_i \colon [0,\overline{Q}] \to \mathbb{R}$ are functions such that for every $x = (x_1 , \ldots ,x_n ) \in S = [0, \overline{Q} ]^N$ player $i$'s payoffs are given by
\begin{equation}
\pi_i(x)=\alpha_i \, H \left( \sum_{j=1}^n h_j(x_j) \right) - g_i(x_i)
\end{equation}
\end{itemize}
 \end{definition}
 
 \noindent
 Throughout we refer to $\Gamma$ as a ``game'' in the same sense as a non-cooperative game defined above and use it synonymously with $G_\Gamma$.
 
All payoffs in a social purpose game depend separately on player $i$'s own strategy, through the term $g_i(x_i)$ that represents an individual cost of player $i$ of executing strategy $x_i$, and on a common term $H\Big( \sum_{i=1}^n h_i(x_i) \Big )$ that represents a common social benefit for all players. Therefore, the function $H$ is naturally interpreted as a \emph{social benefit function}, while for every player $i \in N$ the function $g_i$ represents an \emph{individual cost function}.

The functions $(h_i)_{i \in N}$ can naturally be interpreted as assessments of individual contributions in the aggregated contribution, represented as $\sum_{i \in N} h_i$. If $h_i$ is the identity function for all $i \in N$, we arrive at the standard utilitarian formulation of aggregation. On the other hand, our formulation allows for a wide range of different representations of aggregation through selections of the functions $H$ and $(h_i)_{i \in N}$. We refer to the functions $H$ and $(h_i,g_i)_{i \in N}$ as the \emph{constituting functions} of the social purpose game $\Gamma$.

Furthermore, the payoff of player $i \in N$ depends on the common benefit, weighted by the individual benefit parameter, $\alpha_i >0$, that measures the importance that player $i$ gives to the common benefit $H$ in comparison with the individual cost $g_i$. This allows for asymmetry in the assessment of the common benefit across players as individual contributors to the common benefit. For ease of the following analysis, we assume that all players in $N$ are ranked according to their preference for the generated social benefit, i.e., such that $0< \alpha_1 \leqslant \dots \leqslant \alpha_n$.

A social purpose game $\Gamma$ is clearly an aggregative separable game---applying the insights from \citet{Jensen2013}. The class of social purpose games includes the classes of various well-known aggregative games such as \emph{pollution abatement games} and, more generally, \emph{public good provision games}, which are both characterised by $\alpha_i=1$ for all players $i \in  N$.
\begin{remark}
	 We note that the tragedy of the commons situation considered in Section 2 is indeed a social purpose game. This is made explicit by noting that the payoff function of the tragedy of the commons can be written in the required form through a monotonic transformation. Indeed, taking the (natural) logarithm of her payoff function we can represent the payoff function of user $i \in N$ as $\log \pi_i (x) = \log \left( 1 - \sum_{i \in N} x_i \right) - \left(- \log x_i \right)$. Hence, this corresponds to a social purpose payoff function determined by $\alpha_i =1$, $h_i(x_i)=x_i$, $g_i (x_i) = - \log x_i$ and $H(x) = \log (1-x)$.
\end{remark}

\paragraph{Classes of social purpose games}

In the discussion of social purpose games we have not imposed any properties on the constituting functions $H$ and $(h_i,g_i)_{i \in N}$. It is natural to require these functions to satisfy certain properties. In particular, we consider the case in which $h_i$ are the identity functions, which refers to a class of social purpose games that include the public good provision games. 

This simplification introduces several interesting classes of social purpose games. For these classes of games, we are able to compare the aggregate Nash equilibrium strategies and the aggregate social optimum strategies, as well as the size of the individual contributions in the Nash equilibrium and the social optimum.

\begin{definition}
Let $\Gamma = \langle N , \overline{Q}, H, ( \alpha_i , h_i,g_i)_{i \in N} \rangle$ be a social purpose game, generating the normal form representation $G_\Gamma$.
\begin{numm}
\item The social purpose game $\Gamma$ is called \textbf{regular} if 
\begin{itemize}
\item all functions $h_i$, $i \in N$, are the identity functions $h_i (x_i) = x_i$ and all functions $H$ and $(g_i)_{i \in N}$ are continuously differentiable;
\item  the common benefit function $H$ is continuously differentiable, increasing and concave---implying that its derivative $H'$ is weakly decreasing on $X$, and;
\item the individual cost function $g_i$ is continuously differentiable, increasing and convex in $x_i$ for every $i \in N$---implying that its derivative $g'_i$ is weakly increasing for every $i \in N$.
\end{itemize}
\item The social purpose game $\Gamma$ is called \textbf{strict} if $\Gamma$ is regular and, additionally, for every $i \in N$ the individual cost function $g_i$ is twice differentiable and strictly convex---implying that $g''_i >0$.
\end{numm}
\end{definition}

\noindent
Regularity of social purpose games refers to the concavity of the payoff functions. This implies that regular social purpose games describe a situation in which players make a contribution to a common goal or resource and that these contributions are subject to weakly decreasing returns to scale, individually as well as socially. This class of social purpose games captures a very large number of relevant applications, including environmental problems, the provision of public goods and extractive situations from a common resource.

Strict social purpose games impose that the payoff structure is strictly concave. This additional property implies that one can think of these games as extractive situations from a common resource that is subject to strictly decreasing marginal returns. One can also refer to these strict social purpose games as ``extraction games'' to emphasise the most common and obvious application of this class of games. In particular, we note that the tragedy of the commons is an extraction game that can be represented as a strict social purpose game. 

\subsection{Potentials and Nash equilibria}

In the following proposition we show that each social purpose game is a weighted potential game with weights $\alpha_i$ for all players $i \in N$ \citep{MondererShapley1996}.
\begin{proposition} \label{Prop:Potential1}
Every social purpose game $\Gamma = \langle N , \overline{Q}, H, ( \alpha_i , h_i,g_i)_{i \in N} \rangle$ is a weighted potential game for weight vector $\alpha = ( \alpha_1 , \ldots , \alpha_n )$. The corresponding $\alpha$-potential is given by
\begin{equation}
	P(x)= H \left( \sum_{i=1}^n h_i(x_i) \right) - \sum_{i=1}^n \frac{1}{\alpha_i} \, g_i(x_i)
\end{equation}
\end{proposition}
\begin{proof} 
We prove the result checking that the definition of weighted potential game holds true. For every $x_i, y_i \in X_i$ and $x_{-i}=(x_1, \dots, x_{i-1}, x_{i+1}, \dots, x_n) \in [0, \overline{Q} ]^{n-1}$ we have that:
\begin{align*}
\pi_i(y_i, x_{-i}) & -\pi_i(x_i, x_{-i})= \\[1ex]
& \alpha_iH(h_i(y_i)+ \sum_{j\ne i}h_j(x_j))-g_i(y_i)- \alpha_iH(h_i(x_i)+ \sum_{j\ne i}h_j(x_j))+g_i(x_i)=\\[1ex]
& \alpha_i \Bigl( H(h_i(y_i)+ \sum_{j\ne i}h_j(x_j))-\frac{1}{\alpha_i} g_i(y_i)- H(h_i(x_i)+ \sum_{j\ne i}h_j(x_j))+\frac{1}{\alpha_i}g_i(x_i) \Bigr)= \\[1ex]
& \alpha_i \Bigl( H(h_i(y_i)+ \sum_{j\ne i}h_j(x_j))-\frac{1}{\alpha_i} g_i(y_i)- \sum_{j \ne i}\frac{1}{\alpha_j} g_j(x_j)- H(h_i(x_i)+ \sum_{j\ne i}h_j(x_j))+ \\[1ex]
& \qquad + \frac{1}{\alpha_i}g_i(x_i)+\sum_{j \ne i}\frac{1}{\alpha_j} g_j(x_j) \Bigr)=\alpha_i \left( \, P(y_i,x_{-i})-P(x_i,x_{-i}) \, \right) .
\end{align*}
This concludes the proof of the assertion.
\end{proof}

\paragraph{Properties of Nash equilibria of social purpose games}

We say that social purpose game $\Gamma$ \emph{admits} Nash equilibria if the associated normal form game $G_\Gamma$ has Nash equilibria. Using established insights from potential game theory, we can prove quite straightforwardly the existence of Nash equilibrium for the class of social purpose games. 
\begin{proposition}\label{Prop:NE-exist} 
Consider a social purpose game $\Gamma$.
\begin{abet}
\item If all constituting functions $H$ and $(h_i,g_i)_{i \in N}$ are continuous, then $\Gamma$ admits at least one Nash equilibrium $x^{\mathrm{NE}}=(x_1^{\mathrm{NE}}, \dots, x_n^{\mathrm{NE}}) \in [0, \overline{Q} ]^N$.
\item If the social purpose game $\Gamma$ is strict such that $H$ is non-constant, then $\Gamma$ admits a unique Nash equilibrium.
\item If all constituting functions $H$ and $(h_i,g_i)_{i \in N}$ are continuously differentiable, then every interior Nash equilibrium $x^{\mathrm{NE}} = (x^{\mathrm{NE}}_1 , \ldots , x^{\mathrm{NE}}_n ) \in \left( \, 0, \overline{Q} \, \right)^N$ is a solution to the set of equations given by
\begin{equation}
	H' \left( \, \sum_{j=1}^n h_j(x_j) \right) \, h'_i (x_i) = \frac{g'_i(x_i)}{\alpha_i} \qquad \mbox{for every } i \in N.
\end{equation}
\end{abet}
\end{proposition} 
\begin{proof}
The hypotheses that $H$ is continuous and $h_i$ and $g_i$ are continuous for every $i \in N$, imply that the potential function $H$ is continuous on $[0, \overline{Q} ]^N$. By compactness of $[0, \overline{Q} ]^N$, there exists a maximum of $P$ and, thus, since for a weighted potential game all potential maximisers are Nash equilibria \citep{MondererShapley1996}, i.e., $\arg\max P \subseteq \mathrm{NE}_\Gamma$, where $\mathrm{NE}_\Gamma$ is the set of Nash equilibria of $G_\Gamma$, there exists at least one Nash equilibrium in $G_\Gamma$.
\\[1ex]
To show assertion (b),  for any $s= \sum_{i=1}^n x_i$ let $x_i(s)\geqslant 0$ be the  solution of $ \alpha_i H'(s)=g_i'(x_i(s))$. Since $g_i'$ is strictly monotone, we find $x_i(s)= \bigl(   g_i'   \bigr)^{-1} \bigl( \, \alpha_i H'(s) \, \bigr)$. Since $\Gamma$ is a strict social purpose game and $\bigl(   g_i'   \bigr)^{-1} $ is increasing, the  map $s\to x_i(s)$ is continuous and decreasing on $[0, n \overline{Q}]$.
\\
We define the map $F \colon s\in [0, n \overline{Q}] \mapsto \sum_{i=1}^n x_i(s)\in [0, n \overline{Q}]$, which  is continuous and decreasing on $[0, n \overline{Q}]$, since $\sum_{i=1}^n x_i(0)>0$. Furthermore, under our assumptions $H'(0)>0$, implying that $F$ has a unique fixed point $s^*\in [0, n \overline{Q}] $, i.e.,  $s^*= \sum_{i=1}^n x_i(s^*)$, corresponding to  the unique Nash equilibrium.
\\[1ex]
To show assertion (c), consider an interior Nash equilibrium $x^{\mathrm{NE}} = (x^{\mathrm{NE}}_1 , \ldots , x^{\mathrm{NE}}_n ) \in \big( \, 0, \overline{Q} \, \big)^N$. Then clearly $x^{\mathrm{NE}}$ satisfies the first order conditions of the Nash best response optimisation problem given by
\[
x^{\mathrm{NE}}_i \in {\arg \max}_{0 < x_i < \overline{Q}} \ \  \pi_i \left( x_i , x^{\mathrm{NE}}_{-i} \, \right) .
\]
Noting that
\[
\frac{\partial \pi_i}{\partial x_i} \, \pi_i (x) = \alpha_i \, H' \left( \, \sum_{j=1}^n h_j(x_j) \right) \, h'_i (x_i) - g'_i (x_i) ,
\]
we arrive immediately at the conclusion that assertion (c) is valid.
\end{proof}

\begin{remark}\label{Remark:O1} 
In the case in which the constituting functions are continuously differentiable, any interior Nash equilibrium is a solution of the system of first order conditions given by $\nabla P(x_1, \dots, x_n) =0$. This implies that every interior Nash equilibrium is a stationary point of the potential $P$. This allows the relatively easy computation of the Nash equilibria of a social purpose game.
\\
Note here that given equations in Proposition \ref{Prop:NE-exist}(c) lead to the conclusion that for all $i,j \in N \colon$
\[
\frac{g'_i(x_i^{\mathrm{NE}})}{\alpha_i h'_i(x_i^{\mathrm{NE}})} = \frac{g'_j(x_j^{\mathrm{NE}})}{\alpha_j h'_j(x_j^{\mathrm{NE}})}
\]
Hence, each of these fractions represents a \emph{shadow price} of the individual marginal benefit in terms of the common marginal benefit. \\ If $h_i$ is the identity function for all $i \in N$, this shadow price is simplified to $\rho_i = \frac{g'_i \left( \, x_i^{\mathrm{NE}} \, \right)}{\alpha_i }$. These shadow prices can be used to interpret the tensions between individual costs and benefits in social purpose game situations.
\end{remark}

\noindent
Clearly, by the definition of weighted potential game we have that $P^{\rm max} \subseteq \mathrm{NE}(\Gamma)$, where $P^{\rm max}= \arg\max_{[0, \overline{Q}]^N} P(x)$ is the set of $\alpha$-potential maximizers and $\mathrm{NE}(\Gamma)$ the set of the Nash equilibria of the game $G_\Gamma$.  For a special class of social purpose games, it is possible to characterise the set of the potential maximizers as the set of the Nash equilibria.

\begin{proposition} \label{Prop:Potential2}
Suppose that the social benefit function is linear, i.e., $H(t)=a t+b$ with $a,b\in  \mathbb{R}$. Then it holds that $P^{\rm max} = \mathrm{NE}(\Gamma)$.
\end{proposition}
\begin{proof} 
Under the linear assumption, the potential function is:
\[
P(x)= a \left( \sum_{i=1}^n h_i(x_i) \right) +b  - \sum_{i=1}^n \frac{g_i(x_i)}{\alpha_i}
\]
We must prove the inclusion $\mathrm{NE}(\Gamma)  \subseteq P^{\max}$. Suppose that we have a Nash equilibrium of the game $\Gamma$ denoted by $(x_1^*,...,x_n^*)$ such that for any $i \in N$ and any $x_i \in [0, \overline{Q} ] \colon$
\[
\alpha_i H \left( \sum_{i=1}^n h_i(x_i^*) \right)   -g_i(x_i^*) \geqslant    \alpha_i H  \left(h_i(x_i)  +\sum_{j\neq i} h_j(x_j^*) \right)     -g_i(x_i) 
\]
that is equivalent to 
\[
H \left( \sum_{i=1}^n h_i(x_i^*) \right)   -\frac{g_i(x_i^*)}{\alpha_i } \geqslant     H  \left(h_i(x_i) +\sum_{j\neq i} h_j(x_j^*) \right)     -\frac{g_i(x_i)}{\alpha_i }  
\]
This is under the linear assumption equivalent to 
\[
a\left( \sum_{i=1}^n h_i(x_i^*) \right) +b  -\frac{g_i(x_i^*)}{\alpha_i } \geqslant    a  \left(h_i(x_i)  +\sum_{j\neq i} h_j(x_j^*) \right)  +b   -\frac{g_i(x_i)}{\alpha_i }
\]
Summing up all the inequalities,  we conclude that for all $i \in N \colon$
\[
n a\left( \sum_{i=1}^n h_i(x_i^*) \right) + n b   - \sum_{i=1}^n \frac{g_i(x_i^*)}{\alpha_i }    \geqslant     a\left( \sum_{i=1}^n h_i(x_i) \right) + 
(n-1) a \left(\sum_{i=1}^n h_i(x_i^*) \right)  +n b  - \sum_{i=1}^n \frac{g_i(x_i)}{\alpha_i} 
\]
leading to
\[
a\left( \sum_{i=1}^n h_i(x_i^*) \right)     - \sum_{i=1}^n \frac{g_i(x_i^*)}{\alpha_i }    \geqslant     a\left( \sum_{i=1}^n h_i(x_i) \right) -
   \sum_{i=1}^n \frac{g_i(x_i)}{\alpha_i}
  \]
Taking into account that this holds for all $x \in [0, \overline{Q}]^N$, we conclude that $P(x^*) \geqslant P(x)$ for all $x$, implying that $x^*\in P^{\max}$. This shows the assertion.
 \end{proof}
 
\bigskip\noindent
Let us remark that the linearity of $H$ does not fully characterise the equality proved in Proposition \ref{Prop:Potential2}. Indeed, with reference to Example \ref{Ex:Hnl}, there exist social purpose games in which $H$ is not linear and the potential maximizers coincide with the Nash equilibria of the game. 
 
 Furthermore, we conjecture that the assertion of Proposition \ref{Prop:Potential2} holds for the larger class of social purpose games in which $H$ is concave and piecewise linear. The following example shows this for a specific case.
 
 \begin{example}
 	Consider a social purpose game $\Gamma$ with two players $n=2$ such that $\overline{Q}=1$ and a payoff structure determined by $\alpha_1 = \alpha_2=1$, $h_1(x) = h_2(x) =x$, the common payoff function given by
 	\[
 	H(x) = \left\{
 	\begin{array}{ll}
 		x \qquad & \mbox{for } x \in [0,1] \\
 		1 & \mbox{for } x>1
 	\end{array}
 	\right.
 	\]
 	and $g_1(x) = g_2 (x) =x^2$, resulting in $\pi_1 (x_1,x_2) = H( x_1+x_2) - x^2_1$ and $\pi_2 (x_1, x_2) = H( x_1+x_2) - x^2_2$.
 	\\
 	We note that the unique Nash equilibrium of $\Gamma$ is given by $\left( \tfrac{1}{2} , \tfrac{1}{2} \right)$, which coincides with the unique potential maximiser. 
 \end{example}
 
 \begin{remark}
 	We remark in relation to Proposition \ref{Prop:Potential2} that for every strict social purpose game $\Gamma$ that admits a \emph{unique} Nash equilibrium, it holds that there is a unique potential maximiser which is the unique Nash equilibrium.
 \end{remark}

\subsection{Social optima}

We say that the strategy profile $x^{\mathrm{SO}} \in [0, \overline{Q} ]^N$ is a \emph{social optimum} of the social purpose game $\Gamma = \langle N , \overline{Q}, H, ( \alpha_i , h_i,g_i)_{i \in N} \rangle$ if
\begin{equation}
	x^{\mathrm{SO}} \in \arg \max_{x \in [0, \overline{Q}]^N} \sum_{i \in N} \pi_i (x)
\end{equation}
The next proposition addresses the existence of social optima in social purpose games. For the following analysis it is useful to introduce the \emph{social welfare function} $W \colon [0, \overline{Q} \, ]^N \to \mathbb R$ where
\begin{equation}
	W (x) =  \frac{1}{A} \, \sum_{i=1}^n\pi_i (x) = H \left( \,\sum_{i=1}^n h_i(x_i) \, \right) - \frac{\sum_{i=1}^n g_i(x_i)}{A}
\end{equation}
with $A = \sum_{i \in N} \alpha_i$. Clearly, social optima maximise the function $W$ over $[0, \overline{Q} \, ]^N$. This is used in the proof of the next proposition.

\begin{proposition} \label{Prop:SO-exist}
Consider a social purpose game $\Gamma$.
\begin{abet}
	\item If the common benefit function $H$ and all individual functions $(h_i,g_i)_{i \in N}$ are continuous, then there exists a social optimum $x^{\mathrm{SO}} = \left( \, x^{\mathrm{SO}}_1, \ldots , x^{\mathrm{SO}}_n \, \right) \in  [0, \overline{Q} \, ]^N$.
\item Suppose that $\Gamma$ satisfies the following conditions:
\begin{itemize}
	\item The common benefit function $H$ is increasing and concave;
	\item For every $i \in N \colon h_i$ is concave, and;
	\item For every $i \in N \colon g_i$ is strictly convex.
\end{itemize}
Then there exists a unique social optimum $x^{\mathrm{SO}} \in  [0, \overline{Q} \, ]^N$ for $\Gamma$.
\end{abet}
\end{proposition}
\begin{proof}
Under the stated assumptions, it is obvious that the social welfare function $W$ is continuous on $[0, \overline{Q} \, ]^N$. Hence, applying the Weierstrass theorem, $W$ has a maximum on $[0, \overline{Q} \, ]^N$, which corresponds to a social optimum, showing the assertion.
\\[1ex]
To show assertion (b), note that from the imposed conditions on $H$ and $(h_i,g_i)_{i \in N}$, the constructed social welfare function $W$ is concave on $[0, \overline{Q} \, ]^N$. Moreover, since $g_i$ is strictly convex for every player $i \in  N$, it follows that the social welfare function $W$ is strictly concave. Hence, $W$ has a unique maximiser in $[0, \overline{Q} \, ]^N$, corresponding to the unique social optimum for $\Gamma$.
\end{proof} 

\begin{corollary}
The following properties hold:
\begin{abet}
	\item Every regular social purpose game admits a Nash equilibrium as well as a social optimum.
	\item Every strict social purpose game admits a unique social optimum.
\end{abet}
\end{corollary}
 
\begin{remark} \label{Remark:O2} If all constituting functions are continuously differentiable, the social optimum solves $\nabla W (x_1, \dots, x_n) =0$ implying that for every $i \in N \colon$
\[
H' \left( \, \sum_{j=1}^n h_j(x_j) \, \right) \, h'_i (x_i) = \frac{g'_i(x_i)}{A}
\]
 This compares to the shadow prices derived in Remark \ref{Remark:O1}, indicating the well-established underinvestment for the common benefit in situations like represented by social purpose games.
\end{remark}
 
\begin{remark} \label{Remark:O3} 
For every strategy tuple $x = (x_1, ... ,x_n) \in X \colon P(x) \leqslant W (x)$. 
\end{remark}

\begin{remark}\label{Remark:O4} 
The generic component of the social optimum $x_i^{\mathrm{SO}}$ does not depend on $\alpha_i$ for every $i \in N$.
\end{remark}

\paragraph{Comparing Nash equilibria and social optima in social purpose games}

For the class of social purpose games, the comparison of the Nash equilibria and social optima is usually very informative. The next example considers a very simple payoff structure in which the difference between these two concepts is maximal on the strategy set $[0, \overline{Q} ]$.

\begin{example}
	Consider a regular social purpose game $\Gamma = \langle N , \overline{Q}, H, ( \alpha_i , h_i,g_i)_{i \in N} \rangle$ with the following payoff structure:
	\begin{itemize}
		\item $H \Big( \sum_{i \in N} h_i(x_i) \,\Big) = \sum_{i \in N} x_i$ is a standard utilitarian common benefit function, where $h_i (x_i) = x_i$ is selected as the identity function for every $i \in N$
		\item $g_i (x_i) = x_i$ for every $i \in N$, and;
		\item $\alpha_i <1$ for every $i \in N$ such that $A=\sum_{i \in N} \alpha_i >1$.
	\end{itemize}
	In this example, for any $i \in N$ we can write $\pi_i (x) = \alpha_i \sum_{j \neq i} x_j + ( \alpha_i -1) x_i$ and we see that the unique Nash equilibrium is given by $x^{\mathrm{NE}}_i =0$ for all $i \in N$.
	\\
	On the other hand, the social optimum is identified by maximising the social welfare function
	\[
	\sum_{i \in N} \pi_i (x) = A \sum_{i \in N} x_i -  \sum_{i \in N} x_i  = (A-1) \,  \sum_{i \in N} x_i 
	\]
	which has a unique maximum identified as being given by $x^{\mathrm{SO}}_i = \overline{Q}$ for every $i \in N$.
\end{example}

\noindent
In regular social purpose games, the collective contributions in Nash equilibrium are always lower than the collective contributions required for a social optimum. This refers to the commonly accepted property that is exhibited in public good provision situations.

\begin{proposition}\label{Prop:NE-SO}
Suppose that $\Gamma$ is a regular social purpose game. Then for every Nash equilibrium $x^{\mathrm{NE}}$ and every social optimum $x^{\mathrm{SO}}$ of $G_\Gamma$ it holds that
\begin{equation}
\sum_{i=1}^n x_i^{\mathrm{NE}} \leqslant \sum_{i=1}^n x_i^{\mathrm{SO}}.
\end{equation}
\end{proposition}
 
\begin{proof}
Let us suppose by contradiction that $\sum_{i=1}^n x_i^{\mathrm{NE}} > \sum_{i=1}^n x_i^{\mathrm{SO}}$. Since $H'$ is decreasing in $\sum_{i=1}^n x_i$, then $H' \Big( \sum_{i=1}^n x_i^{\mathrm{NE}} \Big) \leqslant H'\Big( \sum_{i=1}^n x_i^{\mathrm{SO}} \Big)$. Since Remarks \ref{Remark:O1} and \ref{Remark:O2} hold true, it follows that for any given $i \in N \colon \frac{g'_i(x_i^{\mathrm{NE}})}{\alpha_i} \leqslant \frac{g'_i(x_i^{\mathrm{SO}})}{A}$, where $A=\sum_{j=1}^n \alpha_j >0$. 
\\
Since, for any $i \in N$, $0 < \alpha_i < A$ and $g'_i$ is positive, it follows that
\begin{equation}
\frac{g'_i(x_i^{\mathrm{NE}})}{A}<\frac{g'_i(x_i^{\mathrm{NE}})}{\alpha_i} \leqslant\frac{g'_i(x_i^{\mathrm{SO}})}{A} .
\end{equation}
Since $g'_i$ is increasing in $x_i$, we obtain that $x_i^{\mathrm{NE}} \leqslant x_i^{\mathrm{SO}}$. Hence, we have shown that $x_i^{\mathrm{NE}} \leqslant x_i^{\mathrm{SO}}$ for any $i \in N$. This, in turn, implies that $\sum_{i=1}^n x_i^{\mathrm{NE}} \leqslant \sum_{i=1}^n x_i^{\mathrm{SO}}$, contradicting the hypothesis that $\sum_{i=1}^n x_i^{\mathrm{NE}} > \sum_{i=1}^n x_i^{\mathrm{SO}}$. This contradiction proves the assertion.
\end{proof}
  
\begin{remark} \label{Remark:R1}
We can easily generalise the previous result to the case in which functions $h_i$ for all $i \in  N$ are not the identity functions, instead assuming that $\frac{g'_i(\cdot)}{h'_i(\cdot)}$ are increasing functions in $x_i$ for all players $i \in  N$. 
\\
This assumption simply means that an increase of player $i$'s strategy has a greater impact on the individual cost $g_i (x_i)$ than on the common benefit $H \left( \, \sum_i h_i(x_i) \, \right)$ through which player $i$'s contribution is assessed.  
\end{remark}

\noindent
We also note that, if we assume that the common benefit function $H$ is convex, we lose the property stated in Proposition \ref{Prop:NE-SO}.
 
\begin{example}\label{Ex:Hnl}
We consider the non-regular social purpose game $\Gamma$ defined by $n=2$, $\overline{Q} =1$ and payoff functions $\pi_1$ and $\pi_2$ given by $\alpha_1=\alpha_2=\frac{1}{2}$, $H(t)=2t^2-4t$ and $g_1=g_2=0$; then 
\begin{equation}
\pi_1(x_1,x_2)=\pi_2(x_1,x_2)= {(x_1+x_2)^2}-  2(x_1+x_2).
\end{equation}
We identify two Nash equilibria $\{(0,0), (1,1)\}$ that are also social optima. So, the statement in Proposition \ref{Prop:NE-SO} is no longer true.
\end{example}

\noindent
We next compare each component of a Nash equilibrium with the corresponding component of a social optimum.  First, we consider an example of a strict social purpose game in which there are players $i \in N$ with $x^{\mathrm{NE}}_i < x^{\mathrm{SO}}_i$ as well as players $j \in N$ with $x^{\mathrm{NE}}_j > x^{\mathrm{SO}}_j$.

\begin{example}
	Let $\delta > 0$. Consider a strict social purpose game $\Gamma_\delta$ with $N = \{ 1,2 \}$ and $\overline{Q} =1$ with
	\begin{align}
		\pi_1 (x_1,x_2) & = \alpha_1 \log (x_1+x_2 + \delta ) - x^{1+ \delta}_1 \\[1ex]
		\pi_2 (x_1,x_2) & = \alpha_2 \log (x_1+x_2 + \delta ) - x^{1+ \delta}_2
	\end{align}
	where $0< \alpha_1 <1 < \alpha_2$ with $A= \alpha_1 + \alpha_2 <2$ and assuming that $\delta > 0$ is a sufficiently small parameter, implying that the game $\Gamma_\delta$ is indeed a strict social purpose game.
	\\
	We investigate Nash equilibria as well as social optima of this class of strict social purpose games $\Gamma_\delta$ for $\delta >0$ that are sufficiently small.
	\begin{description}
		\item[Nash equilibria:]
		The game $\Gamma_\delta$ has a unique Nash equilibrium that can be determined by investigating the first order conditions for an interior solution. We derive that for $i = 1,2 \colon$
		\[
		\frac{\partial \pi_i}{\partial x_i} = 0 \quad \mbox{implies} \quad \frac{\alpha_i}{x_1+x_2 + \delta} = (1+ \delta ) x^\delta_i
		\]
		For the case that $\delta$ is sufficiently small, $\alpha_1 <1$ and $\alpha_2 >1$, it follows that there is no interior Nash equilibrium and that the only Nash equilibrium is the corner equilibrium given by $x^{\mathrm{NE}}_1 =0$ and $x^{\mathrm{NE}}_2 =1$.
		\item[Social optimum:]
		The social optimum can be determined by optimising the standard social welfare function given by
		\[
		\pi_1 (x_1, x_2) + \pi_2 (x_1, x_2)  = A \log (x_1+x_2 + \delta ) - x^{1+ \delta}_1 - x^{1+ \delta}_2
		\]
		The first order conditions for this optimum can be derived and lead to the conclusion that
		\[
		0< x^{\mathrm{SO}}_1 = x^{\mathrm{SO}}_2 < \left( \, \frac{A}{2(1+ \delta)} \, \right)^{\frac{1}{1+ \delta}} <1
		\]
		due to the property that $A <2$ and that $\delta >0$ is sufficiently small.
	\end{description}
	Hence, in this game we have derived that $x^{\mathrm{NE}}_1 =0 < x^{\mathrm{SO}}_1$ as well as $x^{\mathrm{NE}}_2 =1 > x^{\mathrm{SO}}_2$, while $x^{\mathrm{NE}}_1 + x^{\mathrm{NE}}_2 =1 < x^{\mathrm{SO}}_1 + x^{\mathrm{SO}}_2$.
\end{example}
 It may happen that  $x_i^{\mathrm{NE}} < x_i^{\mathrm{SO}}$ for any $i \in N$, as is shown in the next example.
\begin{example}\label{Ex:SIMPLE}
Consider the simple strict social purpose game $\Gamma^s$ with players $N = \{ 1, \ldots ,n \}$ and $\overline{Q} >0$ sufficiently large.\footnote{Hence, the strategic variable $x_i$ can be assumed as if taken from $\mathbb R_+$ for every $i \in N$. Given the chosen formulation of the payoff structure, it is sufficient to select $\overline{Q} = A = \sum_{i \in N} \alpha_i >0$.} For $x = (x_1 , \ldots , x_n  ) \in \mathbb R^n_+$ the payoff for player $i \in N$ is given by
\begin{equation}
	\pi_i (x) = \alpha_i \sum_{i \in N} x_i - x^2_i \qquad \mbox{with } \alpha_i >0 \quad i \in N .
\end{equation}
Clearly $\Gamma^s$ is a strict social purpose game, which unique Nash equilibrium is given by
\begin{equation}
	x^{\mathrm{NE}}_i = \tfrac{1}{2} \alpha_i \qquad \mbox{resulting in } X^{\mathrm{NE}} = \sum_{i \in N} x^{\mathrm{NE}}_i = \tfrac{A}{2} \mbox{ with } A=\sum_{i \in N} \alpha_i .
\end{equation}
The social optimum is determined as
\begin{equation}
	x^{\mathrm{SO}}_i = \frac{1}{2} A \qquad \mbox{resulting in } X^{\mathrm{SO}} = \sum_{i \in N} x^{\mathrm{SO}}_i = \tfrac{n}{2} A .
\end{equation}
Clearly, $\sum_{i=1}^n x_i^{\mathrm{NE}} \leqslant \sum_{i=1}^n x_i^{\mathrm{SO}}$ and, moreover, $x_i^{\mathrm{NE}} < x_i^{\mathrm{SO}}$ for every $i \in N$.
\\
Now, for every $i \in N$, the Nash equilibrium payoffs are determined as $\pi_i \left( x^{\mathrm{NE}} \right) = \frac{\alpha_i}{4} (2A - \alpha_i ) >0$ and the socially optimal payoffs are given by $\pi_i \left( x^{\mathrm{SO}} \right) = \frac{A}{4} (2 \alpha_i  n-A)$. Note that $\pi_i (x^{\mathrm{SO}}) >0$ if and only if $\alpha_i > \frac{A}{2n}$. Finally, $\pi_i (x^{\mathrm{NE}}) < \pi_i (x^{\mathrm{SO}})$ if and only if $\tfrac{A}{\alpha_i} - \tfrac{\alpha_i}{A} < 2(n-1)$. 
\end{example}

\section{Endogenous emergence of collaboration}

Following the literature on partial cooperation in non-cooperative games developed in \citet{CGL2011,CGL2018,CGL2020} and \citet{MallozziTijs2006,MallozziTijs2007,MallozziTijs2009}, we consider the emergence of stable partial cooperative in social purpose games. We explore particularly the notion of a partial cooperative leadership equilibrium (PCLE) for this class of games prior to developing a stability notion and showing that there indeed might emerge stable partnerships in social purpose games under the PCLE concept. Through the PCLE notion and the definition of stable coalition formation, this incorporates a limited form of farsightedness in the formation of a stable coalition of cooperators.

\subsection{Partial cooperative leadership equilibrium}

We first consider two main notions of partial cooperative equilibrium, following the theory developed in \citet{MallozziTijs2007} and \citet{CGL2011} as well as some of the insights obtained in \citet{CGL2018}.  

Within the context of a social purpose game $\Gamma$ we have assumed throughout that players are ranked by their appreciation of the common benefit, i.e., $0< \alpha_1 \leqslant \alpha_2 \leqslant \cdots \leqslant \alpha_n$. This implies that the higher ranked players have a higher propensity to prefer the common benefit component of their payoffs. Hence, the higher ranked players have an increased preference for the benefit generated through cooperation or collaboration. This leads us to conclude that any cooperation would only emerge among the higher ranked players.

Using this as a fundamental assumption in any cooperation, we denote by the variable $k \in \{ 2,\dots, n\}$ a \textit{level of cooperation}, which refers to the strategic willingness of the $k$ highest ranked players to cooperate and to select a collective strategy for common purpose.\footnote{We remark here that $k=0$ or $k=1$ are meaningless in the context of partial cooperation. Furthermore, the level of cooperation $k=n$ refers to a case of social optimality, since all players in the game act as cooperators. The latter is actually a potential valid level of cooperation.} Consequently, for any level of cooperation $k$, we denote by $C_k=\{n-k+1, \ldots, n\}$ the corresponding coalition of cooperators and by $N_k = N \setminus C_k=\{1, \ldots ,n-k\}$ the complement of $C_k$, being the corresponding set of non-cooperators.

We denote by $x^{C_k} =(x_{n-k+1}, \dots, x_n) \in \prod_{i=n-k+1}^n S_i = [0, \overline{Q} ]^{C_k}$ and $x^{N_k} =(x_{1}, \dots, x_{n-k}) \in \prod_{j=1}^{n-k} S_j = [0, \overline{Q} ]^{N_k}$. In particular, $x^{C_k}$ is collectively determined, while $x^{N_k}$ is competitively selected by individual players $j \in N_k$.

\paragraph{Equilibrium concepts under partial cooperation}

In this setting there naturally result two types of equilibrium concepts to be considered for any given level of cooperation $k \in \{ 2, \ldots ,n \}$. The first one is the \emph{partial cooperative equilibrium} concept---based on the same logic as temporal ``best response'' rationality on which the standard Nash equilibrium concept is founded---in which every non-cooperator $i \in N_k = \{ 1, \ldots , n-k \}$ acts competitively or individually, while all cooperating players $j \in C_k = \{ n-k+1, \dots, n\}$ act cooperatively as a single decision maker. Hence, in this partial cooperative equilibrium concept all decision makers in the set $\{ 1, \ldots , n-k , C_k \}$ act as Nash optimisers. We refer to \citet{CGL2011,CGL2018} for a full analysis of partial cooperative equilibria in general non-cooperative games. Here we limit our discussion to the second notion of equilibrium under partial cooperation.

Alternatively, for any given level of cooperation $k \in \{ 2, \ldots ,n \}$, we might consider the notion of \emph{partial cooperative leadership equilibrium} (PCLE). In such an equilibrium, it is assumed that the coalition of cooperators  $C_k$ acts as a single Stackelberg leader \citep{Stackelberg1934} and all non-cooperators $j \in N_k$ act as Stackelberg followers. In order to explicitly introduce this leader-follower equilibrium concept, for any $x^{C_k} \in [0,\overline{Q}]^{C_k}$ we denote by $\Gamma_{k}(x^{C_k}) = \langle {N_k ,[0,\overline{Q}]^{N_k} , \omega^{x^{C_k}}} \rangle $ the normal form game---called the \emph{conditional partial cooperative game} for $x^{C_k}$---given by player set $N_k = N \setminus C_k$ of non-cooperators whose strategy set is $S_j=[0,\overline{Q}]$ and who have a conditional payoff function $\omega_j^{x^{C_k}}\colon [0,\overline{Q}]^{N_k} \to \mathbb{R}$ defined by
\begin{equation}
	\omega_j^{x^{C_k}} \left( x^{N_k} \right) =\pi_j \left( \, x^{N_k},x^{C_k} \, \right) \qquad j \in N_k = N \setminus C_k .
\end{equation}
We denote by $\mathrm{NE} ({x^{C_k}}) \subset  [0,\overline{Q}]^{N_k}$ the set of Nash equilibria that emerge in this conditional non-cooperative game $\Gamma_{k} (x^{C_k})$. 

Note that, under the assumptions of Proposition \ref{Prop:NE-exist}, $\mathrm{NE} ({x^{C_k}}) \neq \varnothing $ but we cannot guarantee the uniqueness of the Nash equilibrium in $\Gamma_{k}(x^{C_k})$. Thus, in line with \citet{MallozziTijs2007} and \citet{CGL2011}, in order to select one Nash equilibrium among followers, we assume that cooperators are pessimistic in the sense that the coalition of the cooperators, the leader, supposes that the followers' (non-cooperators) choice is the worst for herself and select a maximin strategy. As a consequence, we introduce
\begin{equation}
	\tilde\pi \left( x^{C_k} \right) = \min\limits_{x^{N_k} \in \mathrm{NE} (x^{C_k})} \ \sum_{i \in C_k} \pi_i \left( x^{N_k},x^{C_k} \right)
\end{equation}
and 
\begin{equation}
	\widetilde{X}_k = \left\{ \left. \tilde{x}^{C_k} \in [0,\overline{Q}]^{C_k} \, \right| \, \tilde\pi(\tilde{x}^{C_k})=\max\limits_{x^C \in [0,\overline{Q}]^C} \tilde\pi(x^C) \right\} 
\end{equation}
An action tuple $(x^{N_k}_*, x^{C_k}_*) \in [0,\overline{Q}]^{N_k} \times [0,\overline{Q}]^{C_k}$ is a \textit{partial cooperative leadership equilibrium} for the game $\Gamma$ if $x^{C_k}_* \in \widetilde{X}_k$ and 
\[
x^{N_k}_* \in {\arg\min}_{x^{N_k} \in \mathrm{NE} ({x^{C_k}_*})} \ \sum_{i \in C_k} \pi_i \left( x^{N_k}, x^{C_k}_* \right).
\]
The following result states the conditions under which a PCLE exists in a social purpose game. 

\begin{proposition} \label{thm:PCLE-exist}
	For every regular social purpose game $\Gamma$ and every level of cooperation $k \in \{ 2, \ldots , n \}$ there exists at least one partial cooperative leadership equilibrium.
\end{proposition}

\begin{proof}
	Assume that $\Gamma$ is a regular social purpose game, i.e., $H$ is concave and increasing, $g_i$ are convex and continuous for all $i \in N$ and $h_i (x_i)=x_i$ for all $i \in N$. Take any level of cooperation $k \in \{ 2, \ldots , n \}$ with the corresponding coalition of cooperators $C_k = \{ n-k+1, \ldots ,n \}$.  We first show the following claim:
	\\
	\begin{claim}
		The correspondence ${\bf{E}} \colon [0,\overline{Q}]^{C_k} \to 2^{[0,\overline{Q}]^{N_k}}$ defined by
		\begin{equation}
			 x^{C_k} \in [0,\overline{Q}]^{C_k} \mapsto  {\bf{E}}(x^{C_k})= \mathrm{NE} (x^{C_k}) \subset  [0,\overline{Q}]^{N_k}
		\end{equation}
		is non-empty, compact valued and upper hemi-continuous.
	\end{claim}
	\\
	\textbf{Proof of Claim A:} Since Proposition \ref{Prop:NE-exist} holds, for any $x^{C_k} \in [0,\overline{Q}]^{C_k}$, $\mathrm{NE} (x^{C_k}) \neq \varnothing$, implying that ${\bf{E}}$ is non-empty valued. 
\\
Let $\Bigl (x_p^{C_k} \Bigr )_{p \in \mathbb{N}}$ a sequence such that $x_p^{C_k} \to x^{C_k}  \in [0,\overline{Q}]^{C_k}$. Let us take $x_p^{N \setminus C} \in \mathrm{NE} (x^{C_k}_p)$ such that $x_p^{N_k} \to x^{N_k}$ as $p \to \infty$. Since $x_p^{N_k} \in \mathrm{NE} (x^{C_k}_p)$ we have 
 \begin{align*}
 	\alpha_j H \left( \sum_{j=1}^{n-k} x_{j,p}^{N_k} + \sum_{i=n-k+1}^n x_{i,p}^{C_k} \right) & -g_j \left( x_{j,p}^{N_k} \right) \geqslant\\[1ex]
 	\alpha_j H \left( x_j^{'N_k}+ \sum^{n-k}_{l=1, l\neq j} \right. & \left. x_{l,p}^{N_k} +\sum_{i=n-k+1}^n x_{i,p}^{C_k} \right)-g_j \left( x_{j}^{'N_k} \right)
 \end{align*}
for every non-cooperator $j \in N_k$ and for any $x_{j}^{'N_k} \in [0, \overline{Q}]$.
\\
If $p \to \infty$ by continuity of all functions we obtain 
\begin{align*}
	\alpha_j H \left( \sum_{j=1}^{n-k} x_{j}^{N_k} + \sum_{i=n-k+1}^n x^{C_k}_i \right) & -g_j \left( x_{j}^{N_k} \right) \geqslant\\[1ex]
	\alpha_j H \left(  x_j^{'N_k} + \sum^{n-k}_{l=1, l\neq j} \right. & \left. x_{l}^{N_k}  + \sum_{i=n-k+1}^n x_{i}^{C_k} \right)-g_j \left( x_{j}^{'N_k} \right)
\end{align*} 
for every non-cooperator $j \in N_k$ and for any $x_{j}^{'N_k} \in [0, \overline{Q}]$. Thus, $x^{N_k} \in \mathrm{NE} (x^{C_k})$ leading to the conclusion that ${\bf{E}}$ is closed valued.
\\
Finally, the set  $\mathrm{NE} (x^{C_k})$ is  closed and compact  for any $x^{C_k}  \in [0,\overline{Q}]^{C_k}$, since $\mathrm{NE} (x^{C_k}) \subset [0, \overline{Q}]^{N_k}$ is closed as shown above and $[0, \overline{Q}]^{N_k}$ is compact.  Finally, since every closed correspondence with compact co-domain is upper hemi-continuous, this implies the assertion of the claim. \hfill $\square$
\\[1ex]
From Claim A and the fact that each payoff function $\pi_i$ for $i \in C_k$ and $\pi_j$ for $j \in N_k$ is continuous and quasi-concave, with reference to the proof of \citet[Theorem 2.6]{CGL2011}, there exists at least one PCLE in $\Gamma$. This shows the assertion of Proposition \ref{thm:PCLE-exist}.
\end{proof}

\begin{remark}\label{Remark:O5} 
Denoting by $x^{L}_* = \left( x^{N_k}_* , x^{C_k}_* \right)$ a partial cooperative leadership equilibrium at level of cooperation $k \in \{ 2, \ldots , n \}$ and given $C_k \subset N$ as the corresponding coalition of cooperators, in the same assumptions of Proposition \ref{Prop:NE-SO}, we obtain that $\sum_{i \in C_k} x_i^{\mathrm{NE}} \leqslant \sum_{i \in C_k} x_{*,i}^{L}$. 
\end{remark}
The next example shows explicitly that the mapping $\mathbf{E}$ as constructed in the proof of Proposition \ref{thm:PCLE-exist} does not have to be single-valued. 
\begin{example}\label{nu}
Consider the regular social purpose game $\Gamma$ with $n=2$, $\overline{Q} =1$ and payoff functions $\pi_1$ and $\pi_2$ given by choosing $\alpha_1=\alpha_2=1$, $H(t)=t$, $g_1=x_1, g_2=0$. Then  $\pi_1(x_1,x_2) = {x_2}$ and $\pi_2(x_1,x_2)=  x_1+x_2$.
\\
By considering the player 2 as the only cooperator, for any given $x_2$, the other player maximised $\pi_1(x_1,x_2) = {x_2}$ with respect to $x_1$. The multi-valued map ${\bf{E}}$ of Claim A assigns to every $x_2$ the whole strategy set $[0,1]$. The non-cooperative player 1 determines  her optimal strategy  $\bar x_2=1$   by solving the problem
\[
\max_{x_2}  \ \bigl[ \min_{x_1\in [0,1]}  \ x_1+x_2 \bigr]
\]
Hence, the map $\mathbf E$ is indeed multi-valued.                                     
\end{example}

\subsection{Formation of stable coalitions of cooperators}

Given the existence of a partial cooperative leadership equilibrium, the next objective is to address the following research question: is there a way to endogenously determine the number of cooperators in a partial cooperative framework?

For this, we introduce additional hypotheses to guarantee that a standard vNM stability conception \citep{vNM} can be implemented. In particular, it is required that in the game under consideration, there is a unique PCLE for every level of cooperation. We already established in Proposition \ref{thm:PCLE-exist} that every regular social purpose game admits at least one PCLE for every level of cooperation. Furthermore, for every strict social purpose game the constructed mapping $\mathbf{E}$ in the proof of Proposition \ref{thm:PCLE-exist} is a single-valued function. Therefore, we limit ourselves to the following class of social purpose games throughout the following discussion.

\begin{axiom} \label{ax:uniqueness}
	Throughout the following discussion we only consider strict social purpose games $\Gamma$ that admit a unique Nash equilibrium $x^{\mathrm{NE}}$ as well as a unique partial cooperative leadership equilibrium (PCLE) for every level of cooperation $k \in \{ 2, \ldots ,n \}$.
	\\
	For every level of cooperation $k \in \{ 2, \ldots , n \}$, representing the case that $C _k = \{ n-k+1, \ldots , n \}$ forms as the coalition of cooperators, we denote by $x^L_i (k) \in [0, \overline{Q}]$ the unique PCLE strategy for every $i \in N$ and by $\pi^L_i (k) = \pi_i \left( x^L (k) \, \right)$ the resulting PCLE payoff of player $i \in N$.
\end{axiom}

\noindent
 Under Axiom \ref{ax:uniqueness}, we use the following notion of coalition stability originally proposed for similar decision situations in \citet{Daspremont1983}, which in turn was informed by the vNM stability notion seminally stated in \citet{vNM}. This particular definition of coalitional stability simply means that no player inside the coalition has an incentive to abandon her membership of the coalition of cooperators (``internal'' stability) and no player outside the coalition of cooperators has an incentive to join the coalition of cooperators (``external'' stability).
 
 In the applied formulation below for a level of cooperation $k$, this is only expressed for the two marginal players $n-k \in N_k$ and $n-k+1 \in C_k$ where as before $C_k = \{ n-k+1, \ldots , n \}$ is the corresponding coalition of cooperators under cooperation level $k$. This is founded on the specific role of these marginal players, due to their position in the ranking based on their preference for the common benefit expressed through their respective $\alpha$-weights. This definition has already been applied in the motivating case of the tragedy of the commons as set out in Section 2. Therefore, the formalisation of this notion of stability is as follows.

\begin{definition}
	Consider a strict social purpose game $\Gamma$ that satisfies Axiom \ref{ax:uniqueness}.  We refer to the level of cooperation $k \in \{ 2, \ldots ,n \}$ as
	\begin{numm}
		\item \textbf{Internally stable} in $\Gamma$ if $\pi^L_{n-k+1} (k) \geqslant \pi^L_{n-k+1} (k-1)$, and
		\item \textbf{Externally stable} in $\Gamma$ if $k=n$ or $\pi^L_{n-k} (k) \geqslant \pi^L_{n-k} (k+1)$
	\end{numm}
	where for ease of notation we denote by $x^L (1) = x^{\mathrm{NE}}$. \\
	The level of cooperation $k^* \in \{ 2, \ldots ,n \}$ is \textbf{stable} if $k^*$ is an internally as well as externally stable level of cooperation in $\Gamma$.
\end{definition}
The next result states the necessary and sufficient conditions for stability in the formation of a stable coalition of cooperators in an arbitrary \emph{strict} social purpose game satisfying Axiom \ref{ax:uniqueness}. 

\begin{proposition} \label{prop:gen-stable}
Let $\Gamma$ be a strict social purpose game satisfying Axiom \ref{ax:uniqueness} and consider the level of cooperation $k \in \{ 2, \ldots , n \}$. Denote by $q=n-k+1 \in C_k$ the marginal cooperator and $r=n-k \in N_k$ the marginal non-cooperator.\footnote{We remark that as a consequence $C_{k-1} = C_k \setminus \{ q\}$ and $C_{k+1} = C_k \cup \{ r \}$.} Then the following properties apply:
\begin{numm}
	\item The level of cooperation $k=2$ is stable if and only if 
 		\begin{equation*}
 			\tfrac{1}{\alpha_{n-1}} \left( g_{n-1} \left( x_{n-1}^{L} (2) \, \right) -g_{n-1} \left( x_{n-1}^{\mathrm{NE}} \right) \, \right) \leqslant H \left( \sum_ {i \in N} x_i^{L}(2) \right) -H \left( \sum_{i \in N} x_i^{\mathrm{NE}} \, \right)
 		\end{equation*}
 		as well as
 		\begin{equation*}
 			\tfrac{1}{\alpha_{n-2}} \left( g_{n-2} \left( x_{n-2}^{L} (2) \, \right) -g_{n-2} \left( x_{n-2}^{L} (3) \right) \, \right) \geqslant H \left( \sum_{i \in N} x_i^{L} (2) \right)  -H \left( \sum_{i \in N} x_i^{L} (3) \, \right) .
 		\end{equation*}
	\item For $k \in \{ 3, \ldots , n-1 \}$ it holds that the level of cooperation $k$ is stable if and only if 
		\begin{equation*}
			\tfrac{1}{\alpha_{q}} \left( g_{q} \left( x_{q}^{L} (k) \, \right) -g_{q} \left( x_{q}^{L} (k-1) \right) \, \right) \leqslant H \left( \sum_{i \in N} x_i^{L} (k) \, \right) -H \left( \sum_{i \in N} x_i^{L} (k-1) \, \right)
		\end{equation*}
 		as well as
 		\begin{equation*}
 			\tfrac{1}{\alpha_{r}} \left( g_{r} \left( x_{r}^{L} (k) \, \right) -g_{r} \left( x_{r}^{L} (k+1) \right) \, \right) \geqslant H \left( \sum_{i \in N} x_i^{L} (k) \, \right) -H \left( \sum_{i \in N} x_i^{L} (k+1) \, \right) .
 		\end{equation*}
	\item The level of cooperation $k=n$ is stable if and only if 
 		\begin{equation*}
 			\tfrac{1}{\alpha_{1}} \left( g_{1} \left( x_{1}^{\mathrm{SO}} \, \right) -g_{1} \left( x_{1}^{L} (n-1) \right) \, \right) \leqslant H \left( \sum_{i \in N} x_i^{\mathrm{SO}} \, \right)  -H \left( \sum_{i \in N} x_i^{L} (n-1) \, \right) .
 		\end{equation*}
\end{numm}
\end{proposition}

\begin{proof}
We only state the proof for $k \in \{ 3, \ldots ,n-1 \}$. The proofs for the two extreme cases $k \in \{ 2,n \}$ are omitted, but follow immediately from the reasoning developed below. \\
We compare the payoffs of player $q=n-k+1 \in C_k$ in the stability conditions  if she cooperates with $C_k$ or acts as a non-cooperator $q \in N_{k-1}$. Indeed, if she cooperates, she receives payoff
\begin{align*}
	u_q (C) = \pi^{L}_q (k) & = \alpha_q H \left( \sum_{i \in N} x_i^{L} (k) \right) - g_q \left( x_q^{L} (k) \, \right) \\[1ex]
	& = \alpha_q H \left( \sum_{i \in C_{k-1}} x_i^{L} (k) + x_q^{L} (k) + \sum_{j \in N_k}x_j^{L} (k) \right) - g_q \left( x_q^{L} (k) \, \right)
\end{align*}
and, if she does not cooperate with $C_k$, she would receive
\begin{align*}
	\hspace*{-1em} u_q (NC) = \pi^L_i (k-1) & = \alpha_q H \left( \sum_{i \in N} x_i^{L} (k-1) \, \right) - g_q \left( x_q^{L} (k-1) \, \right) \\[1ex]
	& = \alpha_q H \left( \sum_{i \in C_{k-1}} x_i^{L} (k-1) + x_q^{L} (k-1) + \sum_{j \in N_k}x_j^{L} (k-1) \, \right) - g_q \left( x_q^{L} (k-1) \, \right) .
\end{align*}
Internal stability requires now that $u_q (C) \geqslant u_q (NC)$. This is equivalent to the first condition in the assertion.
\\[1ex]
For the {external stability} condition of the level of cooperation $k \in \{ 3, \ldots , n-1 \}$,  we  compare the payoffs of the marginal non-cooperator $r = n-k \in N_k$ if she cooperates with $C_k$ to form $C_{k+1} = C_k \cup \{ r \}$ or acts as a non-cooperator $r \in N_k$. Indeed, if she cooperates with $C_k$ she receives payoff
\begin{align*}
	u_r (C) = \pi^L_r (k+1) & = \alpha_r H \left( \sum_{i \in N} x_i^{L} (k+1) \, \right) - g_r \left( x_r^{L} (k+1) \, \right) \\[1ex]
	& = \alpha_r H \left( \sum_{i \in C_k} x_i^{L} (k+1) + x_r^{L} (k+1)+ \sum_{j \in N_{k+1}}x_j^{L} (k+1) \, \right) - g_r \left( x_r^{L} (k+1) \, \right)
\end{align*}
and, if she does not cooperate with $C_k$ and remains in $N_k$, she receives
\begin{align*}
	u_r (NC) = \pi^L_r (k) & = \alpha_r H \left( \sum_{i \in N} x_i^{L} (k) \, \right) - g_r \left( x_r^{L} (k) \, \right) \\[1ex]
	& = \alpha_r H \left( \sum_{i \in C_k} x_i^{L} (k) + x_r^{L} (k)+ \sum_{j \in N_{k+1}}x_j^{L} (k) \, \right) - g_r \left( x_r^{L} (k) \, \right)
\end{align*}
External stability now requires that $u_r (NC) \geqslant u_r (C)$, which is equivalent to the second condition of the assertion.
\end{proof}

\paragraph{Stability in the social purpose game of Example \ref{Ex:SIMPLE}}

We determine the partial cooperative leadership equilibrium (PCLE) for any level of cooperation of the simple social purpose game $\Gamma^s$ considered in the Example \ref{Ex:SIMPLE}. Again assume that all players are ranked in order of their preference for the collectively generated benefits with $0< \alpha_1 \leqslant \alpha_2 \leqslant \cdots \leqslant \alpha_n$. In line with the previously applied notation, denote for any non-empty coalition $\varnothing \neq M \subseteq N$ by $A_M = \sum_{h \in M} \alpha_h >0$. Recall that the payoffs are given by
\[
\pi_i (x) = \alpha_i \sum_{i \in N} x_i - x^2_i \qquad \mbox{for every player } i \in N .
\]
and that $x^{\mathrm{NE}}_i = \tfrac{1}{2} \alpha_i$ and $x^{\mathrm{NE}}_i = \tfrac{1}{2} A_N$

Select any level of cooperation $k \in \{ 2, \ldots , n \}$. As before, $C_k = \{ n-k+1 , \ldots , n \} \subset N$ be the coalition of $k$ cooperators with the highest preference for the generated social benefits in this game. The non-cooperators are now the players in $N_k = N \setminus C_k = \{ 1, \ldots , n-k \}$ with the lowest preference for the collectively generated benefit.

Next, we compute the resulting unique PCLE for the level of cooperation $k$. Every non-cooperator $j \in N_k$ selects a best response to all other players' strategies. Hence,
\[
x^{L}_j (k) = \frac{1}{2} \alpha_j \qquad \mbox{for every non-cooperator } j \in N_k .
\]
Note that this best response is independent on the chosen strategy $x_{C_k}$ of the coalition of cooperators $C_k$.

The cooperators in $C_k$ determine their strategies collectively to maximise their collective payoff $\pi_{C_k} = \sum_{i \in C_k} \pi_i$, given $(x^L_j)_{j \in N_k}$. This is again independent of the chosen strategies of all players $j \in N_k$ and results into
\[
x^{L}_i (k) = \frac{1}{2} A_{C_k} = \frac{1}{2} \sum^n_{i = n-k+1} \alpha_i \qquad \mbox{for every cooperator } i \in C_k .
\]
Hence, in the partial cooperative leadership equilibrium we have that
\begin{align}
X^{L} (k) & = \sum_{h \in N} x^{L}_h (k) = \frac{1}{2} A_{N_k} + \frac{k}{2} A_{C_k} \\[1ex]
\pi^{L}_j (k) & = \alpha_j \left( \frac{1}{2} A_{N_k} + \frac{k}{2} A_{C_k} \right) - \frac{1}{4} \alpha^2_j & \mbox{for every } j \in N_k \\[1ex]
\pi^{L}_i (k) & = \alpha_i \left( \frac{1}{2} A_{N_k} + \frac{k}{2} A_{C_k} \right) - \frac{1}{4} A^2_{C_k} & \mbox{for every } i \in C_k
\end{align}
From these payoff levels we deduce that for every non-cooperator $j \in N_k \colon \pi^{\mathrm{NE}}_j < \pi^L_j (k)$ and that for every cooperator $i \in C_k$ it holds that $ \pi^{\mathrm{NE}}_i < \pi^L_i (k)$ if and only if $\tfrac{A_{C_k}}{\alpha_i} - \tfrac{\alpha_i}{A_{C_k}} < 2(k-1)$. 

We next turn to the question whether partial cooperation in this simple social purpose game $\Gamma^s$ results in a stable coalition of cooperators.

\begin{proposition} \label{prop:simple}
Under the partial cooperative leadership equilibrium concept, a level of cooperation $k \in \{ 2, \ldots , n-1 \}$ in the game $\Gamma^s$  is stable if and only if
\begin{equation} \label{eq:simple1}
\alpha_{n-k+1} \geqslant \frac{A_{C_k}}{1+ \sqrt{2(k-1)}}
\end{equation}
as well as
\begin{equation} \label{eq:simple2}
\alpha_{n-k} \leqslant \frac{A_{C_k}}{\sqrt{2k}} .
\end{equation}
Furthermore, $k=n$ is a stable level of cooperation in $\Gamma^s$ to form $C_n=N$ to implement the social optimum if and only if
\begin{equation}
	\alpha_{1} \geqslant \frac{A_{N}}{1+ \sqrt{2(n-1)}}
\end{equation}
\end{proposition}

\begin{proof}
Here, and in the following, for simplicity of notation in formulas, we let $q=n-k+1$ and $r=n-k$ denote  the marginal cooperator $q \in C_k$ and  the marginal non-cooperator $r \in N_k$, respectively. 
\\
 If player $q$ cooperates with $C_k$,  she receives a payoff of
\[
u_q (C) = \pi^{L}_q (k) = \alpha_q \left( \frac{1}{2} A_{N_k} + \frac{k}{2} A_{C_k} \right) - \frac{1}{4} A^2_{C_k}
\]
and if she does not cooperate with $C_k$ she receives
\[
u_q (NC) = \alpha_q \left( \frac{1}{2} A_{N_k} + \frac{1}{2} \alpha_q + \frac{k-1}{2} A_{C_k - q} \right) - \frac{1}{4} \alpha^2_q.
\]
Internal stability requires now that $u_q (C) \geqslant u_q (NC)$. This is equivalent to
\[
\frac{\alpha_q}{2} kA_{C_k} - \frac{1}{4} A^2_{C_k} \geqslant \alpha_q \left( \frac{k-1}{2} (A_{C_k} - \alpha_q) + \frac{\alpha_q}{2}\, \right) - \frac{1}{4} \alpha^2_q 
\]
or
\[
A^2_{C_k} -\alpha^2_q \leqslant 2 \alpha_q A_{C_k} +  2 k \alpha^2_q - 4 \alpha^2_q
\]
or
\[
A^2_{C_k} - 2 \alpha_q A_{C_k} + \alpha^2_q = \left( A_{C_k} - \alpha_q \right)^2 \leqslant 2(k-1) \alpha^2_q
\]
This is equivalent to
\begin{equation}
	A_{C_k} - \alpha_q \leqslant \sqrt{2(k-1)} \, \alpha_q
\end{equation}
which is equivalent to the condition (\ref{eq:simple1}) in the assertion.
\\[1ex]
 If the marginal non-cooperator $r = n-k \in N_k$ cooperates with $C_k$ to form $C_{k+1}$,  she receives payoff
\[
u_r (C) =\alpha_r \left( \frac{1}{2} (A_{N_k} - \alpha_r ) + \frac{k+1}{2} (A_{C_k} + \alpha_r) \, \right) - \frac{1}{4} \left( A_{C_k} + \alpha_r \right)^2
\]
and, if she does not cooperate and remains member of $N_k$, she receives
\[
u_r (NC) = \pi^{L}_r (k) = \alpha_r \left( \frac{1}{2} A_{N_k} + \frac{k}{2} A_{C_k} \right) - \frac{1}{4} \alpha^2_r.
\]
External stability now requires that $u_r (NC) \geqslant u_r (C)$, which is equivalent to
\[
\frac{k}{2} \alpha_r A_{C_k} - \frac{1}{4} \alpha^2_r \geqslant \alpha_r \left( \frac{k}{2} \alpha_r + \frac{k+1}{2} A_{C_k} \, \right) - \frac{1}{4} \left( A_{C_k} + \alpha_r \right)^2
\]
or
\begin{equation}
	A^2_{C_k} + 2 \alpha_r A_{C_k} \geqslant \alpha_r \left[ 2(k+1) A_{C_k} + 2k \alpha_r - 2k A_{C_k} \, \right] = \alpha_r \left( 2 A_{C_k} + 2k \alpha_r \, \right)
\end{equation}
which is equivalent to $A^2_{C_k} \geqslant 2k \alpha^2_r$ and, hence, (\ref{eq:simple2}).
\\
Finally, $k=n$ is a stable level of cooperation if and only if $u_1 (C) \geqslant u_1 (NC)$. This is equivalent to $A_N - \alpha_1 \leqslant \sqrt{2(n-1)} \alpha_1$ as derived above.
\end{proof}

\singlespace
\bibliographystyle{ecta}
\bibliography{RPDB}

\end{document}